\newif\ifLETTER
\LETTERfalse 

\ifLETTER
\documentclass[twocolumn,11pt]{IEEEtran}
\else
\documentclass[11pt]{article}
\usepackage[top=2cm, bottom=2cm, left=2cm, right=2cm, heightrounded,
  marginparwidth=2.9cm, marginparsep=2mm]{geometry}
\fi

\usepackage{soul}
\usepackage{amsmath,amsfonts,amssymb,mathrsfs,epsfig,bm}
\usepackage{amsthm}
\usepackage{mathtools}
\usepackage{soul,url}
\usepackage{graphicx}
\usepackage{subfig}
\usepackage{comment}

\ifLETTER
\else
\fi
\usepackage[usenames, dvipsnames]{xcolor}
\usepackage[colorlinks=true]{hyperref}

\newtheorem{lemma}{Lemma}
\newtheorem{corollary}{Corollary}
\newtheorem{proposition}{Proposition}

\def\Ints{\mathbb{Z}}

\def\P{\mathbb{P}}
\def\E{\mathbb{E}}

\def\BB{\mathcal{B}}
\def\PP{\mathcal{P}}

\renewcommand{\phi}{\varphi}
\renewcommand{\epsilon}{\varepsilon}

\newcommand{\eqdist}{\stackrel{\text{\rm (d)}}{=}}

\definecolor{mygray}{gray}{0.4}
\definecolor{deeppink}{RGB}{255,20,147}
\definecolor{mygreen}{rgb}{0.0, 0.75, 0.0}
\definecolor{myred}{rgb}{0.768, 0.09, 0.09}

\long\def\symbolfootnote[#1]#2{\begingroup
\def\thefootnote{\fnsymbol{footnote}}\footnote[#1]{#2}\endgroup}

\newcounter{para}

\definecolor{eqcol}{RGB}{255,10,130}

\usepackage[margin=0cm]{caption}
\usepackage{microtype}
\DeclareCaptionFont{ls}{\lsstyle} 
\ifLETTER
\else

\fi
\usepackage{caption}

\def\uu{\mathfrak u}

\begin{document}

\title{Age of information distribution under
dynamic service preemption
}

\author{
\begin{tabular}{ccc}
\sc George Kesidis &\sc   Takis Konstantopoulos & \sc Michael A.\ Zazanis\\
\small School of EECS &\small 
\small Department of Mathematics &\small  Department of Statistics\\ 
\small Pennsylvania State University&\small  University of Liverpool &\small  Athens University of Economics \& Business\\
\small \href{mailto:gik2@psu.edu}{\tt gik2@psu.edu} 
&\small 
\href{mailto:takiskonst@gmail.com}{\tt takiskonst@gmail.com} 
&\small 
\href{mailto:zazanis@aueb.gr}{\tt zazanis@aueb.gr}
\end{tabular}
}

\ifLETTER
\else
\date{\small \today}
\fi
\maketitle

\begin{abstract}
Age of Information (AoI) has emerged as an important
quality-of-service measure for applications
that prioritize delivery of the freshest information,
e.g., virtual or augmented reality over mobile devices
and wireless sensor networks used in the control of
cyber-physical systems.
We derive the Laplace transform of
the stationary AoI 
for the M/GI/1/2 system with a ``dynamic" service
preemption and pushout policy depending on the existing
service time of the in-service message.  
Thus, our system generalizes both the static 
M/GI/1/2 queue-pushout system without service preemption
and the M/GI/1/1 bufferless system with service preemption
- two systems considered to provide very good AoI performance.
Based on our analysis,
for a service-time distribution that is
a mixture of deterministic and exponential, we 
numerically show that
the dynamic policy has lower mean AoI than that of these
two static policies and also that of the 
well studied M/GI/1/1 blocking system.
\end{abstract}

\section{Introduction}

Consider a queueing system transmitting messages, particularly
where the service-time distribution models access and transmission
delays over a wireless channel.
If $A(t)$ is the maximum arrival time of messages which are completely 
served before time $t$, then the quantity $\alpha(t)=t-A(t)$
is called
``age of information" (AoI), see 
\cite{Sun19,Yates21} and references therein. 
The reason is simple: in several applications it is the freshness of information
that is important rather than the correct transmission of all messages.
Examples include  virtual reality and online gaming on mobile devices,
semi or fully autonomous vehicles, 
and wireless sensors of power systems and other ``cyber physical" systems.

Mean AoI results for stationary
systems are obtained in, e.g., \cite{Yates12,Costa16,Kosta17,KKZ21b}.
However, in such latency critical applications,
bounds on the tail of the AoI {\em distribution},
e.g., \cite{Inoue19,KKZ19,KKZ21a},
(not just its mean) need to be met.
That is,  for a threshold $\nu$ and tolerance $\epsilon$
such applications require $\P(\alpha>\nu)<\epsilon.$
(Note that, abusing notation, $\alpha$ will typicall stand for a random variable
that is distributed according to $\alpha(t)$ for some, and hence all, $t$,
when the process is stationary.)

Service preemption, last-in-first-out (LIFO) queueing, or
queue push-out  is often not practically feasible,
but blocking/dropping of arriving messages typically is.
The stationary distribution of AoI under blocking or
push-out policies has been derived for different queueing models,
including different buffer sizes and under preemptive or non-preemptive
service, in \cite{Inoue19,KKZ19,Champati19,KKZ21a,KKZ21b}.
Typically, renewal models of interarrival time and service time have 
been considered in prior work.

In particular,
$\PP_1$  (a system with at most one message and  preemptive push-out service)
has pathwise equal AoI as the infinite-buffer LIFO 
system with service preemption.
Similarly, $\PP_2$ (a system with at most two messages and
non-preemptive service with push-out of the queued message) 
has pathwise smaller AoI than $\BB_2$
(a blocking, non-preemptive system with at most two messages).
When service times are deterministic, $\PP_2$ was
shown to have lower mean AoI than $\PP_1$, 
though the
converse is true when service times are exponential
\cite{Kosta17}.
For deterministic service, $\BB_1$  (non-preemptive service with
at most one message in the system) has lower mean AoI
than $\PP_2$ for sufficiently large traffic loads under
deterministic service. Ignoring practical constsraints on queueing
policy,  prior work has identified the policies
$\PP_1$, $\PP_2$ or $\BB_1$ as performing optimally by some AoI measure,
particularly minimizing the mean AoI in steady-state.
One can show $\BB_n$ has pathwise smaller AoI than $\BB_{n+1}$
for $n\geq 2$, however the same statement has not been proved
for LIFO $\PP_n$ policies, e.g., \cite{KKZ21b}.

In this paper, we derive the AoI distribution 
for the stationary M/GI/1/2 system with a ``dynamic" service
preemption or queue-pushout policy depending on the  amount of
service received so far by the in-service message.   As such, it is
a causal queueing policy that 
generalizes both $\PP_1$ and $\PP_2$.
The approach conditions on a well-known Markov-renewal
embedding 
\cite{Cinlar75,Kleinrock75} 
which can be employed to
compute the AoI distribution for other such queueing systems
\cite{KKZ21a}, particularly
LIFO systems with larger buffer sizes 
\cite{KKZ21b} and systems under  the GI/M model
 (renewal arrivals and exponential service times).

We numerically show that for a particular serive-time distributions in the 
M/GI case that
the dynamic policy has lower mean AoI than 
$\PP_1$, $\PP_2$ or $\BB_1$.
Our aim is to draw attention to an open problem: 
For a given queueing system model, what queuing and service policies
minimize an AoI-based quality-of-service measure?
Since for exponential service times, the memoryless property
imples that the bufferless pushout system $\PP_1$ has pathwise minimal
AoI,
we specifically
exclude consideration of memoryless service times in the following.

This paper is organized as follows.
In section 2, we define our queueing system and give some preliminary results.
In section 3, we derive the AoI distribution of the stationary queuing system
under consideration.
Numerical comparisons with other systems based on mean AoI
are made in Section 4.
Finally, we summarize and discuss future work in Section 5.

\section{System Definition and Preliminaries}  \label{sec:prelim}

There is a buffer consisting of two cells. Cell 1 is
reserved for the message receiving service
and cell 2 for the message waiting. 
If there is a message in cell $1$ at time $t$ we let $\uu(t)$ 
be the amount of service received by this message up to $t$;
if the system is empty, we set $\uu(t)=0$.
Fix $\theta \in [0, \infty]$.
If a message arrives at time $t$ and $\uu(t) \le \theta$ then the arriving message 
pushes-out the message in cell 1 and takes its place.
Otherwise, if $\uu(t) > \theta$ then the arriving message 
occupies cell 2 (pushing out the message sitting there, if any).
We call this system $\PP_{2,\theta}$.
Note that $\PP_{2,0}$ and $\PP_{2,\infty}$ make sense too and that
the collection
$\PP_{2,\theta}$, $0\le \theta \le \infty$, is a ``homotopy'' between
these two systems. In fact, in the terminology of \cite{KKZ21a,KKZ19}, 
$\PP_{2,0} = \PP_2$ and $\PP_{2,\infty}=\PP_1$. 
(In the latter system, cell 2 will never be occupied, so, effectively, it
has buffer of size 1.)

Thus, a contiguous service interval  that ends with a message departure is  a
sequence of preempted message-service periods followed
by a completed/successful message-service period.
Prior to the successful message-service period,
there are no queued messages waiting for service.
During the successful message-service period, 
any arriving messages obviously fail to preempt
and, under queue pushout, the {\em last}
 such arriving message is queued and
begins service once the successful message-service
period ends.

We assume that the arrival process is Poisson with rate $\lambda>0$ and that
messages have i.i.d.\ service times (independent of arrivals) distributed like
a random variable $\sigma$ such that $\sigma  >0$ a.s.\ with expectation
$1/\mu< \infty$. We let $G$ be the distribution function of $\sigma$
and set $\hat G(s) = \E e^{-s \sigma}$.
Under these assumptions, we will assume that the system $\PP_{2,\theta}$ is
in steady-state (taking into account that there is a unique such steady-state, the
reasons for which are classical and will not be discussed here).

Let $S_n$ denote successful departure epochs.
Let $K_n$ be the number of mesages in the system immediately
after  $S_n$. 
Given $K_{-1}=0$, consider Figure \ref{fig:0}  at bottom.
\begin{figure}[h]
\begin{center}
\includegraphics[width=2in]{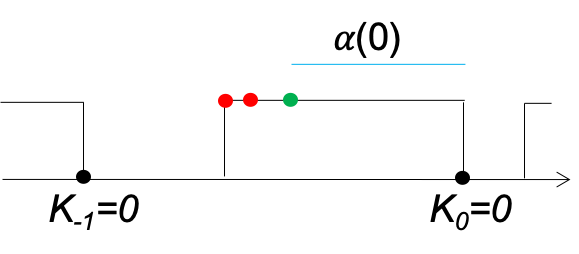}
\hspace*{5mm}
\includegraphics[width=2in]{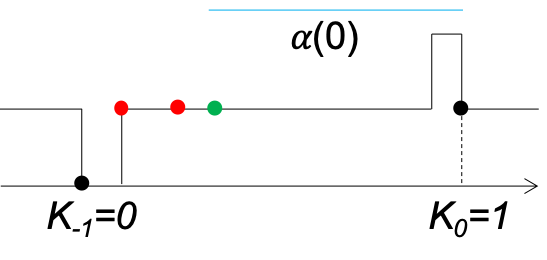}
\captionof{figure}{Number of messages in the system (black line)
when $K_{-1}=0,K_0=0$ 
(top) 
and  $K_{-1}=0,K_0=1$ 
(bottom).
A black dot indicates a (successful) message departure. 
A green dot indicates an arrival 
of a message that will
successfully depart. 
A red dot indicates an arrival 
of an unsuccessful message.
If the fate of an arrival is not indicated in the figure, then
it has no indicating dot.
The length of the blue line is $\alpha(S_0)$.}
\label{fig:0}
\end{center}
\end{figure}

A message service period is successful with probability  
\begin{align*}
1-q:=\P(\tau>\theta \wedge \sigma) 
& = \int_0^\infty e^{-\lambda (\theta\wedge s)} dG(s)\\
& = (1-G(\theta))e^{-\lambda \theta} +\int_0^\theta e^{-\lambda s} dG(s).
\end{align*}
So, considering the successful service period which concludes at $S_0$:
\begin{align*}
\P(K_0=0|K_{-1}=0)
 & = \P(\tau>\sigma~|~\tau>\theta\wedge\sigma)\\
 & = (1-q)^{-1}\hat{G}(\lambda)\\
& =1-\P(K_1=1|K_{-1}=0).
\end{align*}

Given $K_{-1}=1$, consider Figure \ref{fig:11}.
\begin{figure}[h]
\begin{center}
\includegraphics[width=1.6in]{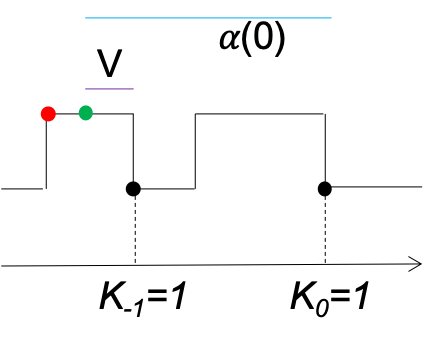}
\hspace*{5mm}
\includegraphics[width=2in]{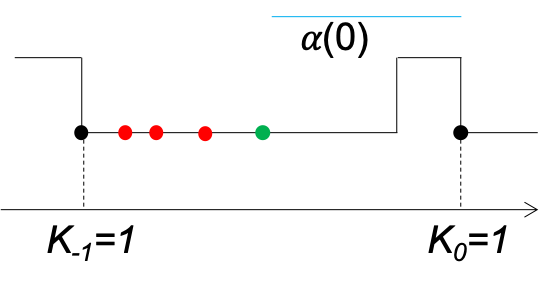}
\captionof{figure}{$K_{-1}=1,K_0=1$ cases.}
\label{fig:11}.
\end{center}
\end{figure}

Use the memoryless property of interarivals
to similarly obtain
\begin{align*}
\P(K_0=0|K_{-1}=1) & =
(1-q)^{-1}\hat{G}(\lambda)\\
	& =1-\P(K_0=1|K_{-1}=1)
\end{align*}
So, $K_n$  is an i.i.d.\ Bernoulli sequence with
\begin{align*}
	p_0 & := \P(K_n=0)  =
(1-q)^{-1}\hat{G}(\lambda) \\
	& = 1-\P(K_n=1) =:  1-p_1.
\end{align*}

The queueing process over consecutive intervals
$[S_{i-1},S_i)$ and
$[S_i,S_{i+1})$ are conditionally independent given $K_i$.
Thus, $\{S_i,K_i\}_{i\in\Ints}$ is 
Markov-renewal with renewal times $S_i$ and
the 
queueing process is semi-Markov 
\cite{Cinlar75,Kleinrock75}.
In particular,  $\alpha(S_i)$ and $S_{i+1}-S_i$ are conditionally independent
given $K_i$.

\section{Stationary AoI Distribution}

\begin{proposition}\label{prop:main}
The Laplace transform of the 
stationary AoI distribution  is
\ifLETTER
$\E e^{-s \alpha(0)}  = $
	{\small
\begin{align}\label{palm}
	\frac{
\sum_{i=0}^1 \E^0(e^{-s\alpha(0)}|K_0=i)
\Big(1- \E^0(e^{-s(S_0-S_{-1})}|K_{-1}=i) \Big)p_i}
	{s\E^0(S_1-S_0)},
\end{align}
	}
\else
\begin{align}\label{palm}
\E e^{-s \alpha(0)} & = 
	\frac{
\sum_{i=0}^1 \E^0(e^{-s\alpha(0)}|K_0=i)
\Big(1- \E^0(e^{-s(S_0-S_{-1})}|K_{-1}=i) \Big)p_i}
	{s\E^0(S_1-S_0)},
\end{align}
\fi
where $\P^0$ and $\E^0$ are respectively probability and expectation given $S_0=0$.
\end{proposition}
\begin{proof}
The Palm inversion formula
$\E e^{-s \alpha(0)} =$\\
$ { \E^0 \int_{S_0}^{S_1}e^{-s \alpha(t)}dt}/{\E^0 (S_1-S_0)}
$
has numerator
\begin{align*}
 & = \E^0 \int_{S_0}^{S_1}e^{-s (\alpha(S_0)+t-S_0)}dt\\
  &  = \E^0 \int_{S_0}^{S_1}
	\sum_{i=0}^1 (e^{-s\alpha(0)}|K_0=i) \E^0(e^{-s(t-S_0)}| K_0=i) p_i dt
\\
  &= \sum_{i=0}^1 \E^0(e^{-s\alpha(0)}|K_0=i) 
	\E^0\Bigg(\frac{1-e^{-s(S_1-S_{0})}}{s}\Bigg| K_0=i\Bigg) p_i
\end{align*}
\end{proof}

To calculate the terms in \eqref{palm} we need to follow the steps
outlined in the lemmas below.
\ifLETTER  
Define $\E^0_{i,j}(\cdot)=\E^0(\cdot|K_{-1}=i,K_0=j)$.
\fi
Let $\hat{F}_0(s)= \hat{G}(s+\lambda)/\hat{G}(\lambda)$ and
$\hat{J}(s)  = \int_0^\infty e^{-s y} dJ(y)$ where
$J(y)= 1$  for $y\geq \theta$  and, for $0\le y<\theta,$
\begin{align*}
dJ(y)  &  = q^{-1} \lambda e^{-\lambda y}(1-G(y)) dy
\end{align*}

\begin{lemma}
\ifLETTER
\begin{align}
& \E^0_{0,0} e^{-s(S_0-S_{-1})} =
\frac{\lambda}{\lambda+s}\cdot\frac{1-q}{1-q\hat{J}(s)}
\hat{F}_0(s)
\label{expS_00}
\\
& \E^0_{0,0} e^{-s \alpha(S_0)}  = \hat{F}_0(s) 
\label{expA_00} 
\end{align}
\else
\begin{align}
& \E^0(e^{-s(S_0-S_{-1})} \mid K_{-1}=0,K_0=0)  =
\frac{\lambda}{\lambda+s}\cdot\frac{1-q}{1-q\hat{J}(s)}
\hat{F}_0(s)
\label{expS_00}
\\
& \E^0(e^{-s \alpha(S_0)} \mid K_{-1}=0,K_0=0) = \hat{F}_0(s) 
\label{expA_00} 
\end{align}
\fi
\end{lemma}
\begin{proof}
See Figure \ref{fig:0} at 
top 
and consider the interval $[S_{-1},S_0)$.
Let $\tau_{-1} $ be  first message arrival time in this interval minus 
$S_{-1}$, so that $\tau_{-1}\sim\exp(\lambda)$ by the memoryless property.
Note that there is a geometric  number $N$ of 
interarrival times
 each of which is  smaller than {\em both} $\theta$ and 
the associated service time; $N=2$ in Figure \ref{fig:0}.
The probability of such unsuccessful service is
\[
\P(\tau<\theta\wedge\sigma) = q.
\]
So, $\P(N=k)=(1-q) q^k$ for $k=0,1,2,...$.  
Let $Y \eqdist (\tau | \tau < \theta\wedge \sigma)$ 
so that $\P(Y\leq y)=J(y)$.
Finally, let
$\tau_0\sim\exp(\lambda)$ be the duration between the arrival 
time (green dot) of the message that departs at $S_0$ and
the next arrival time.
The service time 
(from the green dot to $S_0$)
$\sigma_0$  is independent of $\tau_0$.
Considering the  prior $N$ unsuccessful service completions,
we are given that $\tau_0 >\sigma_0$ or 
$\tau_0>\theta$. Given $K_0=0$, $\tau_0>\sigma_0$.
Let $X_0 \eqdist (\sigma_0|\tau_0>\sigma_0)$ which has distibution
\[
dF_0(x)  =  \hat{G}(\lambda)^{-1} e^{-\lambda x} dG(x), \quad x > 0,
\]
with $\E e^{-s X_0}=\hat{F}_0(s)$. So,   
$$(S_0-S_{-1}\mid K_{-1}=0,K_0=0) \eqdist 
\tau_{-1} + \sum_{n=1}^N Y_n + 
X_0,$$
a sum of independent terms with $Y_n \eqdist Y$.
Also
$\alpha(S_0) \eqdist X_0$ in this case. 
\end{proof}

Define
\begin{align*}
dF_1(x) & = \frac{(e^{-\lambda\theta}-e^{-\lambda x}) dG(x)}
{\int_\theta^\infty (e^{-\lambda\theta}-e^{-\lambda z}) dG(z)} ,
\quad x>\theta\\
\hat F_1(s) & =  \int_\theta^\infty e^{-sx}  dF_1(x).
\end{align*}

\begin{lemma}
\ifLETTER
\begin{align}
& \E^0_{0,1} e^{-s(S_0-S_{-1})} =
\frac{\lambda}{\lambda+s}\cdot \frac{1-q}{1-q\hat{J}(s)}\hat{F}_1(s)
\label{expS_01}
\\
& \E^0_{0,1} e^{-s \alpha(S_0)}
 = \int_\theta^\infty e^{-sx}  dF_1(x) =: \hat{F}_1(s).
\label{expA_01}
\end{align}
\else
\begin{align}
& \E^0(e^{-s(S_0-S_{-1})} \mid K_{-1}=0,K_0=1)  =
\frac{\lambda}{\lambda+s}\cdot \frac{1-q}{1-q\hat{J}(s)}\hat{F}_1(s)
\label{expS_01}
\\
& \E^0(e^{-s \alpha(S_0)} \mid K_{-1}=0,K_0=1) = \hat{F}_1(s).
\label{expA_01}
\end{align}
\fi
\end{lemma}
\begin{proof}
See Figure \ref{fig:0} at bottom. 
The difference between this and
the previous case is that here $\sigma_0>\tau_0>\theta$. 
So, the distribution of  $X_1 \eqdist (\sigma_0\mid \sigma_0>\tau_0>\theta)$ is
$dF_1(x)$, as defined above.
\end{proof}

Define 
\begin{align*}
dH(v)   & = \frac{\lambda e^{-\lambda v}(1-G(v+\theta)) dv}{
 \int_\theta^\infty (1-e^{-\lambda(x-\theta)}) dG(x)}, \quad v \ge 0,\\
\hat{H}(s) & = \int_0^\infty e^{-sv} dH(v).
\end{align*}

\begin{lemma}
\ifLETTER
\begin{align}
& \E^0_{1,0} e^{-s(S_0-S_{-1})}  =
\frac{1-q}{1-q\hat{J}(s)}\hat{F}_0(s)
\label{expS_10}
\\
&\E^0_{1,0} e^{-s \alpha(S_0)} 
 =   q\hat{F}_0(s) +(1-q)\hat{H}(s)\times\eqref{expS_10}
  \label{expA_10}
\end{align}
\else
\begin{align}
& \E^0(e^{-s(S_0-S_{-1})} \mid K_{-1}=1,K_0=0) =
\frac{1-q}{1-q\hat{J}(s)}\hat{F}_0(s)
\label{expS_10}
\\
& \E^0(e^{-s \alpha(S_0)} \mid K_{-1}=1,K_0=0) =  
 q\hat{F}_0(s) +(1-q)\hat{H}(s)\times\eqref{expS_10}
  \label{expA_10}
\end{align}
\fi
\end{lemma}
\begin{proof}
See Figure \ref{fig:10}.
\begin{figure}[h]
\begin{center}
\includegraphics[width=1.6in]{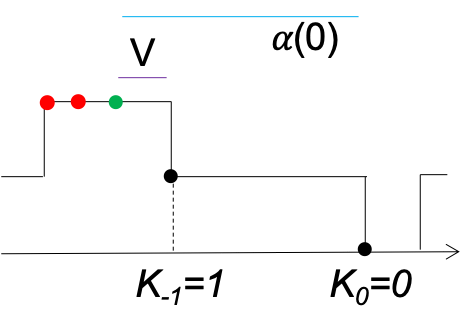}
\hspace*{5mm}
\includegraphics[width=2in]{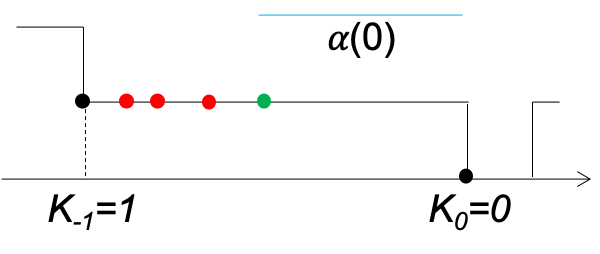}
\captionof{figure}{$K_{-1}=1,K_0=0$ cases.}
\label{fig:10}.
\end{center}
\end{figure}
For $\alpha(S_0)$ 
there are two subcases depending of 
whether there are initial unsuccessful arrivals in the interval
$[S_{-1},S_0)$, i.e., whether $N>0$.
When $N>0$ (Figure \ref{fig:10} bottom), 
this case is like 
when $K_{-1}=0,K_0=0$.
Otherwise (Figure \ref{fig:10} top), 
\[
\alpha(S_0) =V+ S_{0}-S_{-1},
\]
where $V$ is the duration between $S_{-1}$ 
the last arrival before $S_{-1}$ 
 (green dot), and $V$ and $S_0-S_1$ are independent given $K_{-1}=1$. 
Starting from $S_{-1}$, look backward in time until the first Poisson point appears (green dot) and condition on the event that this occurs 
at least $\theta$ units of time before the service time ends.
Thus, $V \eqdist (\tau|\tau<\sigma-\theta,\sigma>\theta)$  with distribution
$H$.
Also, $S_0-S_{-1}$
is distributed as for the  case where $K_{-1}=0,K_0=0$
except 
the first interarrival time is absent.
\end{proof}

\begin{lemma}
\ifLETTER
\begin{align}
&\E^0_{1,1} e^{-s(S_0-S_{-1})}  =
\frac{1-q}{1-q\hat{J}(s)}\hat{F}_1(s)
\label{expS_11}\\
&\E^0_{1,1} e^{-s \alpha(S_0)}
=   q\hat{F}_1(s) +(1-q)\hat{H}(s)\times\eqref{expS_11}
  \label{expA_11}
\end{align}
\else
\begin{align}
&\E^0(e^{-s(S_0-S_{-1})} \mid K_{-1}=1,K_0=1) =
\frac{1-q}{1-q\hat{J}(s)}\hat{F}_1(s)
\label{expS_11}\\
&\E^0(e^{-s \alpha(S_0)} \mid K_{-1}=1,K_0=1)
=   q\hat{F}_1(s) +(1-q)\hat{H}(s)\times\eqref{expS_11}
  \label{expA_11}
\end{align}
\fi
\end{lemma}
\begin{proof}
See Figure \ref{fig:11}. 
For $S_0-S_{-1}$, this case is a combination of the previous two cases,
and for $\alpha(S_0)$ follow 
the previous case except use
\eqref{expA_01} instead of \eqref{expA_00}.
\end{proof}

\noindent\textbf{The final stage:}
The formulas obtained in the lemmas above must now be substituted into \eqref{palm} 
as follows:
\begin{align}
&\E^0(e^{-s(S_0-S_{-1})}|K_{-1}=0)  = 
\eqref{expS_01}\times p_1 + \eqref{expS_00}\times p_0 \label{ES_0}\\
&\E^0(e^{-s(S_0-S_{-1})}|K_{-1}=1)  = 
\eqref{expS_11}\times p_1 + \eqref{expS_10}\times p_0  \label{ES_1}.
\end{align}
\ifLETTER
\[
\E^0 (S_1-S_{0}) 
=-\frac{d(p_1\times\eqref{ES_1}+p_0\times\eqref{ES_0})}{ds}\Bigg|_{s=0}.
\]
\else
So,
\[
\E^0 (S_1-S_{0})
=-\frac{d}{ds} \E^0 e^{-s (S_0-S_{-1})} \Bigg|_{s=0}
=-\frac{d}{ds}  (p_1\times\eqref{ES_1}+p_0\times\eqref{ES_0})\Bigg|_{s=0}.
\]
\fi
\ifLETTER
\else
Moreover,
\fi
\begin{align*}
\E^0 (e^{-s\alpha(0)}|K_0=0) & =
p_0\times \eqref{expA_00} + p_1\times \eqref{expA_10} 
\\
\E^0 (e^{-s\alpha(0)}|K_0=1) & = 
p_0\times \eqref{expA_01} + p_1\times \eqref{expA_11}.
\end{align*}

\section{Numerical Results for Stationary Mean AoI, $\E\alpha(0)$}

The stationary mean AoI can be obtained from \eqref{palm}
and numerically minimized over $\theta$ for a given set of model parameters
for an arbitrary service-time distribution $G$. 
For example, for exponential service times, $dG(x) = \mu e^{-\mu x} dx$,
and $\theta=\infty$ ($\PP_1$) achieves minimal 
$\E \alpha(0)=\mu^{-1}+\lambda^{-1}$ \cite{Kosta17,Inoue19,KKZ19}.
For another example,
for constant service time  $dG(x)=\delta_{1/\mu}(x)dx$,
the $\PP_2$ system  (i.e., when $\theta=0$)  has
\ifLETTER
\begin{align*} 
	\E \alpha(0) &=
\frac{1}{\mu} \left( \left(1 -e^{-\rho}\right)\left(1+\frac{1}{\rho}\right)  \right.\\
	& ~~~~~+ \left. \frac{1}{\rho^2+ \rho e^{-\rho}} \, \left(e^{-\rho}+ \rho e^{-\rho} + \frac{1}{2}\rho^2 \right) \right),
\end{align*}
\else
\begin{equation*} 
\E \alpha(0) \;=\; \frac{1}{\mu} \left( \left(1 -e^{-\rho}\right)\left(1+\frac{1}{\rho}\right) + \frac{1}{\rho^2+ \rho e^{-\rho}} \, \left(e^{-\rho}+ \rho e^{-\rho} + \frac{1}{2}\rho^2 \right) \right);
\end{equation*}
\fi
which is generally smaller than $\PP_1$ 
(and also smaller than $\BB_1$ for sufficiently small traffic loads 
\cite{Inoue19,KKZ19,KKZ21a}).

\begin{figure}[h]
  \begin{center}
\ifLETTER
    \includegraphics[width=\columnwidth]{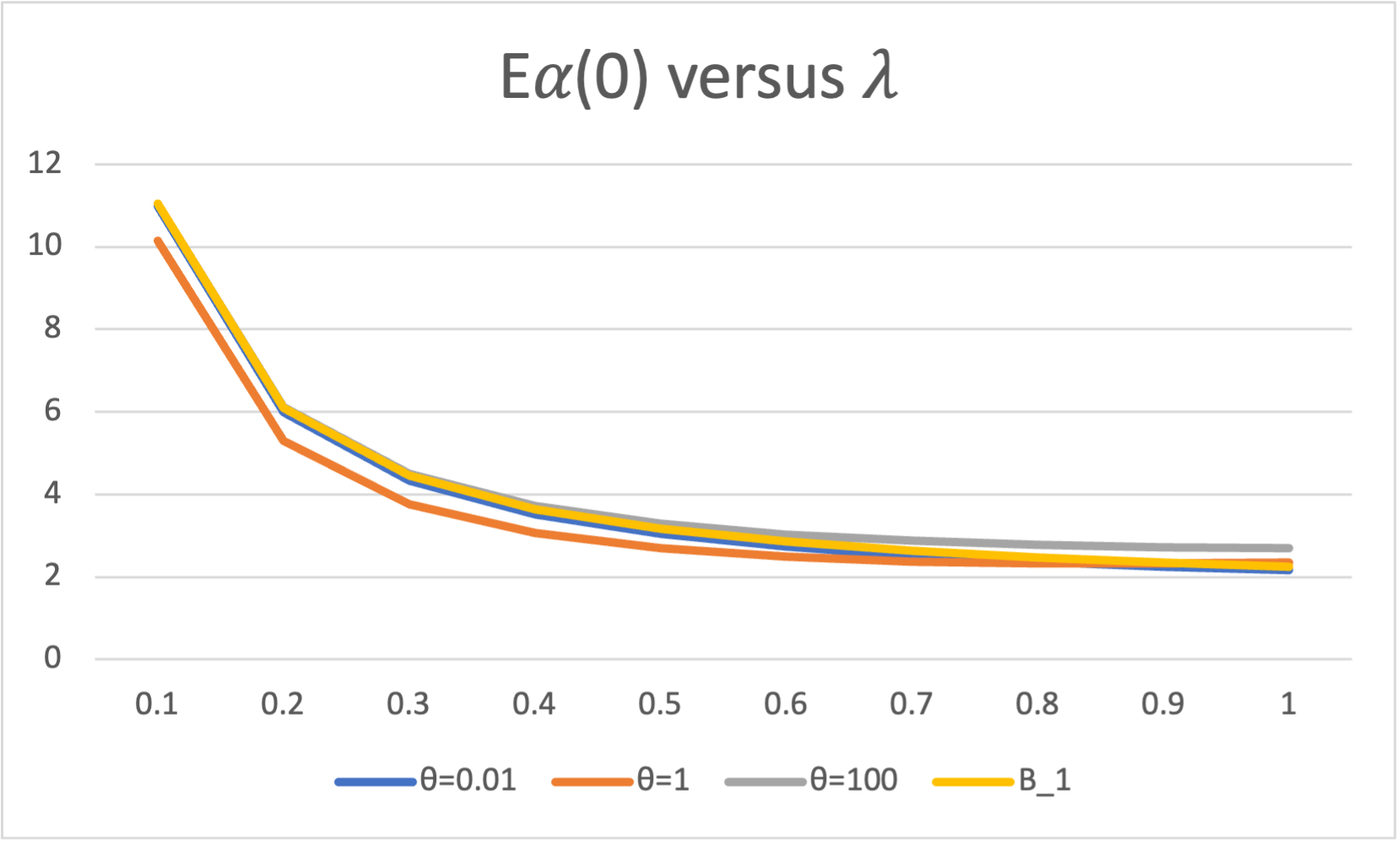}
\else
    \includegraphics[width=4in]{numer/dynamic-mix.png}
\fi
\captionof{figure}{$\E \alpha(0)$ for service-time
distribution which is a  mixture of 
deterministic and exponential with mean $\mu^{-1}=1$.}
  \label{fig:mix}
  \end{center}
\end{figure}
Now consider the  mixture service-time distribution defined by
$dG(x) = \frac12 \delta_1(x)dx +\frac12 e^{-x}dx$ for $x>0$, and so
$\mu=1$. From Figure \ref{fig:mix}, which was obtained
numerically using \eqref{palm}, we see that  in some cases
$\E\alpha(0)$ is {\em not} minimized at either $\theta=0$
or $\theta=\infty$.  That is, in some cases (specifically,
traffic loads $\leq 0.8$), the $\PP_{2,\theta}$ policy
for a finite $\theta>0$ ($\theta =1$ in Figure \ref{fig:mix})
has lower mean AoI  than $\PP_1$, $\PP_2$ and $\BB_1$. 

\section{Summary and Future Work}

In this paper, we analyzed the Age of Information
performance of the $\PP_{2,\theta}$ queueing
policy in steady state under the M/GI model
(Poisson message arrivals and i.i.d. service times). 
$\PP_{2,\theta}$ can hold at most two messages in the 
system and dynamically employs either queue push-out 
or service preemption depending on whether the system is
full and on the service-time-so-far of the in-service message.
We numerically demonstrated
that it has lower mean AoI than 
previously studied $\PP_2$, $\PP_1$ and $\BB_1$ 
policies for a service time distribution that
is a mixture of deterministic and exponential.

Using the Markov embedding approach we employed,
one can also analyze variations of the
$\PP_{2,\theta}$ policy in the same way. 
For example, consider
the policy where a message arriving at time $t$ 
does not preempt the in-service message (but joins the queue in cell 2)
if $0<\uu(t) < \theta$, otherwise the message captures the server.

In future work, 
we will consider the problem of 
determining a queuing and service policy that is optimal with respect to
an AoI based quality-of-service metric 
when interarrival time and service time models
are specified.
In practice, this question may subject to
technological constraints which may preclude use of,
e.g., service preemption and/or queue push-out.
Also, we will consider the problem of 
deriving the AoI distribution for specific 
interarrival time and service time models of
the non-renewal type.

\bibliographystyle{plain}
\bibliography{../../latex/AoI,../../latex/stochastic}

\end{document}